\documentclass [11pt]{article}
\usepackage[dvips,pagebackref,colorlinks=true,linkcolor=blue, citecolor=blue]{hyperref}
\usepackage{amsmath}
\usepackage{amssymb}
\usepackage{algorithmic}
\usepackage{algorithm}
\usepackage{wrapfig}

\usepackage[left=1in,top=1in,right=1in,bottom=1in,nohead]{geometry}

\def\note#1{}

\newcommand{\gnote}[1]{{\textbf{[Grant Note:} #1 \textbf{]}}}

\def\epsilon{\varepsilon}
\def\phi{\varphi}

\newlength{\boxedparwidth} 
\newenvironment{model}[1]{\begin{center}\begin{tabular}{|@{\hspace{0.15in}}c@{\hspace{0.15in}}|}
\hline \\
\begin{minipage}[t]{\boxedparwidth}
\setlength{\parindent}{0.25in}
\noindent\textbf{#1 Assumptions.}
}%
{\end{minipage} \\ \\ \hline \end{tabular}\end{center}}

\newtheorem{theorem}{Theorem}

\newtheorem{remark}{Remark}
\newtheorem{definition}[theorem]{Definition}
\newenvironment{proof}{\noindent {\sc Proof:}}{$\Box$ \medskip}
\DeclareMathOperator\E{\mathbb{E}}

\begin{document}

\title{Finding Overlapping Communities in Social Networks: Toward a Rigorous Approach}

\author{Sanjeev Arora\thanks{Princeton University, Computer Science Department and Center for Computational Intractability {\tt arora@cs.princeton.edu}. Supported by National Science Foundation.}  \and Rong Ge\thanks{Princeton University, Computer Science Department  and Center for Computational Intractability  {\tt rongge@cs.princeton.edu}.Supported by National Science Foundation.} \and Sushant Sachdeva\thanks{Princeton University, Computer Science Department  and Center for Computational Intractability {\tt sachdeva@cs.princeton.edu}. Supported by National Science Foundation.} \and Grant Schoenebeck\thanks{Princeton University, Computer Science Department  and Center for Computational Intractability {\tt gschoene@cs.princeton.edu}.  This research was supported by the Simons Foundation
    Postdoctoral Fellowship.}}
\maketitle

\begin{abstract}

A {\em community} in a social network is usually understood to be a group of nodes more
densely connected with each other than with the rest of the
network. This is an important concept in most domains where networks arise: social,
technological, biological, etc. For many years algorithms for finding communities implicitly
assumed communities are nonoverlapping (leading to use of clustering-based
approaches) but there is increasing
interest in finding overlapping communities. A barrier to  finding
communities  is that the solution concept is often
defined in terms of an NP-complete problem such as Clique
or Hierarchical Clustering.

This paper seeks to
initiate a rigorous approach  to the problem of finding overlapping communities, where ``rigorous'' means that we clearly state the following: (a) the object sought by our algorithm (b) the assumptions about the underlying network (c) the (worst-case) running time.

Our assumptions about the network lie between worst-case and
average-case. An average-case analysis would require a precise
probabilistic model of the network, on which there is currently no consensus.
However, some plausible assumptions about network parameters can be
gleaned from a long body of work in the sociology community spanning five decades focusing on the study of individual communities and {\em ego-centric networks} (in graph theoretic
terms, this is the subgraph induced on a node's neighborhood).  Thus our
assumptions are somewhat ``local'' in nature. Nevertheless   they suffice to
permit a rigorous analysis of running time of algorithms that recover global structure.

Our algorithms use random sampling
similar to that in property testing and algorithms for dense graphs.
We note however that our networks are not necessarily dense graphs, not even in
local neighborhoods.

Our algorithms explore a local-global relationship between ego-centric and
socio-centric networks that we hope will provide a fruitful framework
for future work both in computer science and sociology.

\end{abstract}

\newpage

%This work studies the problem of identifying communities in a graph representing a social network when each node may belong to more than one community,  but at most some constant $d$, communities.
%We present a polynomial time algorithm that, given a graph which is the union of cliques of a particular size where each vertex is contained in at most $d$ cliques will identify all the cliques (with some minor assumptions).  Next we relax the assumption that all the communities are the same size.  We then explore the setting where the communities are random dense graphs.  We note that this is a generalization of a special case dealt with in the work by McSherry \cite{McSherry01}.
%Finally, we present a generalization to the general case of McSherry again, allowing each element to be the part of $d$ partitions, but require that the partitions only have overlap $d^2/10$.

\section{Introduction}
\label{sec:intro}

{\em Community structure} is an important characteristic of
social networks and has long been studied in sociology. The classic paper of Luce and Perry in 1949 ---which introduced the term ``Clique'' to graph theory ---described a community as subsets of individuals every pair of whom are acquainted.  The text of Scott~\cite{Scott94} equates communities with objects such as {\em cliques} or other {\em dense subgraphs}.  Another seminal 1974 paper, Breiger~\cite{Breiger74} develops a theory of communities in terms of \emph{affiliation networks}, which in graph theoretic terms consist of using a bipartite graph with people on one side and communities on the other. In sociology today, the answer to many natural and important questions depends on a better understanding of community structure. Can you travel from one node to another random node using only a few ``strong ties" \cite{Granovetter82}?  Do networks contain ``wide" bridges \cite{CentolaM07}? How much do communities overlap?

The problem of identifying communities arose independently in other fields such as  internet search, study of the web graph, and the problem of clustering network nodes (in networks of biological interactions, citations, etc.). In his recent comprehensive survey of algorithmic approaches, Fortunato~\cite{fortunato}
divides them into two camps based upon whether or not the algorithm assumes ---implicitly or explicitly--- that communities are disjoint.
Assuming disjointness implicitly leads to a view of a community as a \emph{nonexpanding node set}: it contains many edges but has relatively few edges leaving it~\footnote{The Girvan-Newman~\cite{GirvanN02} algorithm does not explicitly define communities as nonexpanding sets. Instead it defines the {\em betweenness} of a node $u$ as the fraction of nodes $v, w$ whose shortest path passes through $u$. It iteratively removes nodes of low betweenness to isolate communities.}. This viewpoint suggests many approaches that have been tried: graph partitioning, hierarchical clustering, spectral clustering, simulated annealing, modularity, betweenness, etc.
Gibson et al.~\cite{GibsonKR98} discovered interesting communities via hubs and authorities; Hopcroft et al.~\cite{Hopcroftetal03} used agglomerative clustering on the the Citeseer database and exhibited interesting communities that persist over time.

However, recently Leskovec et al~\cite{LeskovecLDM08-jounal} presented an extensive study of many of the above methods on larger datasets, and question whether they uncover meaningful structure at larger scales. Leskovec et al. often detect a large ``core" in the network that is difficult to break into communities. One possible interpretation is that if there are communities in the core, they must overlap.

Thus there is growing interest in finding communities that are allowed to overlap, as they do in most real-life social networks. When communities overlap, each community will not in general be a nonexpanding set. (Consequently, clustering-based approaches may not work.) For instance, imagine that the network contains many communities that are equal-sized cliques with bounded pairwise intersections, and every person belongs to four communities. Then each community/clique will have in general as much as three times as many edges going out of it as are contained in it.

Approaches for finding overlapping communities involve either heuristic clique-finding, or local-search procedures that maintain overlapping clusters and improving them via a series of heuristics. Sometimes a probabilistic generative model is assumed and a max-likelihood fit is attempted via EM and other ideas. (The very recent survey of Xie et al.~\cite{Xiesurvey} evaluates dozens of competing heuristics introduced in the last two years alone.)  However, Fortunato states at the end of his  100-page survey:

{\em
..research on graph clustering has not yet given a satisfactory solution of the problem and leaves us with a number of important open issues. The field has grown in a rather chaotic way, without a precise direction or guidelines...What the field lacks the most is a theoretical framework...everybody has his or her own idea of what a community is.}

\paragraph{Modeling communities.} Of course, it is entirely possible that Fortunato's questions have no clean answer, or at least one that spans all types of networks of interest. Quite possibly, a community in a network of gene-gene interactions is an inherently different object than one in the graph of facebook friendships.
Furthermore, a clear definition of the problem does not in itself guarantee a simple algorithm---e.g., if the definition involves cliques.

A related issue is  development of models for community formation/growth whose mathematical analysis yields predictions testable
on real-life networks. An inspiration here is the large body of
Barabasi-Albert~\cite{BarabasiA99} style models which make predictions about degree distributions, graph distance etc. One concrete attempt to model community formation
 is Lattanzi and Sivakumar's~\cite{LattanziS-09} {\em affiliation networks} model that is inspired by sociology work. In this dynamic model the communities are cliques (or dense subgraphs), and the mode of network growth is adding of either new individuals (i.e., nodes) or new communities (i.e., cliques) to an affiliation network. New individual partially copy the community memberships of existing individuals, and new communities are offshoots (subsets) of existing communities. Additional generative
models with community structure appear in~\cite{McSherry01,AiroldiMFX08,BorgsCDL10,KumarRRSTU00}.

What Fortunato's questions point to, though, is a seeming chicken-and-egg situation. Developing models requires reliable data about community structure in real networks. Conversely, finding reliable data about community structure requires some implicit model, since without such a model the algorithm is consigned to solving worst-case instances of NP-hard problems like Clique, Dense Subgraph, and Small-set Expansion. (This issue clearly does not arise for simpler graph properties like node degrees.)

\paragraph{This paper.} We seek a more rigorous approach to the problem of finding communities, where ``rigorous'' means that we clearly state the following: (a) the object sought by our algorithm (b) the assumptions about the underlying network (c) the (worst-case) running time.  We try to break out of the chicken-and-egg situation as follows. Instead of proposing a generative model {\em per se},
we list fairly minimal assumptions about the network that are based on theoretical and empirical work in sociology.  Moreover, these assumptions are ``local" in nature and they largely depend upon objects well-studied in sociology, namely ego-centric networks and individual communities.  We think these assumptions will be satisfied by many  plausible generative models (including Latanzi-Sivakumar, according to our simulations~\cite{rabinovich}). Thus it is interesting that these suffice to recover the communities.
%It is interesting that these minimal local assumptions already suffice to recover communities, then the global %structure of those communities ---e.g., pairwise intersections, growth over time, etc.--- could be studied to aid %the design of an accurate generative model in future.

%We do not attempt an accurate generative model.

Since our formalization of (a) and (b) draws on sociology we expect our approach to apply more to, say, the Facebook graph than biological networks. Furthermore, since our algorithms involve random sampling in node neighborhoods, they may mesh well with  a dominant approach for network study in sociology, namely, {\em ego-centric analysis}~\cite{Wellman07}.   An \emph{ego-centric} network~\cite{Scott94} consists of a person (the ego) and his ties (called ``alters").  Ego-centric networks and their structure have been extensively studied in sociology (often via questionnaires and field-study)~\cite{Wellman07} as a way to gain insight into the entire network~\cite{Wellman79,Burt92,Burt04,Fischer82}. They give a view of how people develop and manage their social network resources~\cite{MarinW-09,Wellman79}.  Data about such networks is easy to collect, even in the field, far from computers~\cite{HoganCW07}. The are even included on the biyearly General Social Survey, a central resource in sociology~\citation{GSS,Burt84,Marsden87}.

\paragraph{Sociological Foundations for our Assumptions.}
Sociologists have observed that ego-centric networks can be clustered into a few communities, from which they infer that individuals participate in only a small number of communities~\cite{WellmanHBBCCKKT05,Hogan10,Burt92,Burt04}.
%
%That these networks of alters form a small number of cluster has long been observed~\cite{Wellman07,Hogan10,SpencerP06,PahlS04}, and this very fact is often central to %analysis of the ego centric networks (for example in Burt's analysis of structural holes~\cite{Burt92,Burt04}).  This observation has been put to pragmatic use in %eliciting these networks from individuals~\cite{HoganCW07,Fitzgerarld78,McCartyG05} by helping people list more ``alters" by grouping the ``alters" into communities.  That ego-centric networks can be clustered into a small number of communities, also reinforces the intuition that these communities overlap.  The running time of our algorithms will depend strongly on the number of communities that each node is involved.
%Not only can ego-centric networks be clustered into a few communities,
Furthermore, they have observed that a large portion of a person's ties fit into communities~\cite{WellmanHBBCCKKT05,Hogan10}.  The celebrated theory of Feld~\cite{Feld81} gives a theoretical understanding for this fact based upon ``foci.'' He defines a \emph{focus} as ``a social, psychological, legal, or physical entity around which joint activities are organized (e.g., workplaces, voluntary organizations, hangouts, families, etc.)" and his theory says that they are responsible for creating many ties in the network.

Mathematically, one could say that each individual participates in up to $d$ communities, and  these communities explain $\gamma$ fraction of his/her ties. This
information already may greatly help the algorithm, whose running time may depend upon $d$ and $\gamma$.

What is a community? As mentioned, in sociology communities are thought of as either cliques or as dense subgraphs which are ``relatively tightly connected" together compared with the rest of the network (see chapter 6 of~\cite{Scott94} for an introduction and survey of how sociology models communities).  Sometimes variants are considered, e.g., Alba~\cite{albaclique} considers cliques of the $t$th power of the graph (i.e., edges correspond to having distance $\leq t$ in the original graph). Jackson's  text~\cite{Jackson08} allows communities to be dense graphs and describes various models for how edges are generated within a community: e.g., $G(n,p)$ or the {\em expected degree model}.

Furthermore, not all dense graphs would pass muster as communities~\cite{Scott94,FriggeriCF11}. For example, a union of two disjoint cliques of size $t$ is a fairly dense graph of size $2t$ but the latter is not community.   We will assume that a ``community'' is a dense subgraph with edges inside it generated according to the expected degree model.  (While this assumption simplifies the exposition, in Section~\ref{subsec:relax1} we will observe that this assumption can be relaxed to deal with other families of dense graphs.)  However, we leave it as future work to extend our notions to more hierarchical notions of communities such as in~\cite{WhiteBB76}.

Another principle often used in sociology is {\em maximality}: we should not be able to add nodes to the community and get the same structure \cite{FriggeriCF11}, otherwise these nodes should be considered part of the community.  (This was the basis for the maximal clique problem introduced in Luce and Perry's 1949 paper; see also the text of
Scott\cite{Scott94}.) Thus nodes within the community are (in some way) better attached to the community than nodes outside of the community.

\subsection{Formal Assumptions and Statement of Results}
\label{subsec:assumptions}

Our assumptions are grounded in the above observations.
%Instead we propose a setting that is halfway between worst-case and average-case in which we make outcome guarantees provide that given network satisfies some mild assumptions which incorporate intuition for the social networks literature.
The network is a
%random
 graph of  size $n$. Each edge $(u,v)$ has a probability $e_{u,v}$ with which it is picked. This can even be
$1$, so we allow adversarial edges.
Each community $C$ is an {\em arbitrary} subset of nodes (unknown to the algorithm), and each node is in at most $d$ communities, so that the communities are allowed to overlap.  We
% call $d$ the \emph{maximum community overlap}, and we
think of $d$ as constant or small (though several of our algorithms run in polynomial time even for $d =2^{\sqrt{\log n}}$).  Any edge $(u, v)$ where $u$ and $v$ share a community $C$ is a \emph{community edge}.  The remaining edges we call \emph{ambient edges}.

%\medskip

\noindent{\bf Assumption 1)} {\em Community edges are chosen according to the expected degree model.}

Each node $u$ in $C$ has an {\em affinity} $p_u$ that lies in $[0,1]$. (Node $u$ has a different affinity for each
community it belongs to.)  Two nodes $u, v$ in $C$ are connected by an edge with probability
$p_u p_v$ (this is called the {\em expected degree model} in the standard text Jackson~\cite{Jackson08}).  Notice that the  model is sufficiently flexible to include  other well-studied cases: the subcase $p_u=1$ for all $u\in C$ corresponds to ``$C$ is a clique,'' and the subcase $p_u =\sqrt{\alpha}$
for all $u \in C$ corresponds to ``$C$ is a dense subgraph generated according to the random graph model $G(k, \alpha)$.''
We will usually assume $p_u$ is lowerbounded by some constant, so that the community is always a fairly dense graph. However in Section~\ref{sec:sparse} we show conditions under which our algorithms can handle even sparser communities.

For nodes $u$, $v$ that belong to more than one community, the probability that they are connected is at least the maximum of $p_up_v$ for all the communities they are in.

While Assumption 1 seems to be tending towards
a ``generative model," we use this formulation primarily for ease of presentation.  We remark in Section~\ref{subsec:relax1} that our algorithms work so long as the edges are ``well-distributed.''

%\medskip

\noindent{\bf Assumption 2)} {\em Maximality assumption with gap $\epsilon$ (also called ``Gap Assumption'').}

Nodes outside the community $C$ are less strongly connected to it than community nodes are.
For example, if one posits  that $p_u = \sqrt{\alpha}$ for $u\in C$ ---i.e., each $u$ has edges to about $\alpha$
community nodes---then our assumption would say that each $w \not \in C$ has edges to {\em less than} a $\alpha -\epsilon$ fraction of nodes in $C$.
This seems reasonable since otherwise one  should consider whether $w$ should belong to $C$ as well. Of course, such assumptions about {\em maximality}
are standard, though the gap $\epsilon$ in this context is new. We are able to relax the Gap Assumption in certain instances (see Section~\ref{subsec:relax2}), though the results are no longer as clean and crisp.

% definitely a computer science touch.

%\medskip
\noindent{\bf Assumption 3)} {\em Community membership accounts for a significant portion of each node's edges, say a constant fraction $\gamma>0$.}

%Assumption  comes directly from observation 2. And finally, we note that $d$ is often in the exponent of the running time, and %so we need this to be small for our algorithms to be efficient.  Observation 3 assures us that this will be the case.

%\medskip

Surprisingly, Assumptions $1$-$3$  suffice to let us efficiently recover the communities
even though we made {\em no other assumptions} about the ambient edges, and allow arbitrary affinities.

\medskip

\noindent{\bf Informal Theorem 1} (See Theorem~\ref{thm:anysizedense}) {\em If all node affinities are lowerbounded by $\sqrt{\alpha}$ then the communities can be recovered in $n^{C \log kd}$
time, where $C=C_{\alpha, \gamma, \epsilon}$ and $k$ is size of the largest community.}

%\medskip

Unfortunately, this running time is only ``quasipolynomial'' instead of polynomial, since $\log k$ could be as high as $\log n$.  We can get polynomial and even near-linear time algorithms for more restricted versions of the problem.
The paper contains many such theorems and the following is representative.

\iffalse
\medskip
\noindent{\em Assumption 4: Communities are fairly distinct.}
Roughly speaking, for each node $u$ in community $C$, at least a constant factor (say $0.1$) of $C$ does not lie in any other community containing $u$. This seems a fair assumption from a sociological viewpoint, since communities are believed to arise in networks because they provide utility to their members above and beyond what is obtainable in existing communities.
\fi

% when affinities and arbitrary community sizes can be solved under
\medskip

\noindent {\bf Informal Theorem 2} (See Section~\ref{sec:similarsize})  {\em For all constants $\alpha, \delta>0$, if all affinities are $\sqrt{\alpha}$ (note: this includes cliques as a special case), and all communities sizes are within a 
$1/\delta$ factor of each other, then the communities can be recovered in time $O(n k d^{C \log d})$ where $C=C_{\alpha, \delta, \gamma, \epsilon}$. Moreover, if affinities are only guaranteed to be at least $\sqrt{\alpha}$, the communities can be recovered in time $O(n k^{2\log(10/\epsilon)/\alpha+1} d^{C\log d})$.}

Our approach makes heavy use of the Gap Assumption and uses random sampling coupled with some exhaustive enumeration of small
cliques in the subgraph induced on the sample.
% (As hinted, the Gap Assumption is not crucial in many cases. In its absence,
%the $\epsilon$ corresponds to imprecision in the answer.)
% , {\em the neighborhood of each node is a dense graph}, (meaning it contains a quadratic number of all possible edges) and thus one can make %progress using sampling techniques that are
Similar ideas are well-known in property testing~\cite{GoldreichGR98} and the related field of approximation algorithms for NP-hard problems in dense graphs~\cite{AroraKK95,FriezeK99}.
Note however that our graphs are not dense. If $d=O(1)$ then the induced graph in the neighborhood of each node is  dense, though we do allow $d$ to be large as $2^{\sqrt{\log n}}$ in many settings. Even when $d=O(1)$ we do not know how to use, for example, the weak regularity lemma~\cite{FriezeK99} (a standard tool in those other fields) to recover the communities.

%Our algorithms will, for example, take a constant-size random sample of a node's neighborhood and then run brute force algorithms (or any other %suitable heuristic) on the induced graph to derive insight into the global structure. Note that the entire graph is not dense.

The use of local sampling in the neighborhood of a single node gives our algorithms the feel of {\em ego-centric} analysis in sociology.  However, when we examine communities as dense subgraphs, our algorithms (partially) explore a two-hop neighborhood from a starting node, which is a more generalized ego-centric analysis, and is necessary because no one node has ties to the entire community.  We also adapt our ideas to the case where communities are not dense graphs (as is plausible in really large networks). Though our results are preliminary, that algorithm explores a two-hop neighborhood from a starting node.

%\gnote{Here would be a great place for a summary of results, perhaps in a grid format}

%\gnote{We should make remarks about sparse section here.  Introduce it more in its %section}

\paragraph{Paper Organization.}
Section~\ref{sec:similarsize} presents algorithms for the case when all community sizes are within an $O(1)$ factor of each other. These algorithms are quite efficient and are a good
introduction to our techniques.
Section~\ref{sec:differentsize} allows communities to have vastly different sizes, derives the most general result (Theorem~\ref{thm:anysizedense}, our Informal Theorem 1), and then
studies how to derive more efficient algorithms for specialized cases or under additional assumptions.

Section~\ref{sec:real-life} shows how in certain cases some of the Assumptions 1-3 can be relaxed. Here there are many open problems, which are also discussed in
% made throughout that members outside communities have noticeably less connection to the community than members inside the community.
Section~\ref{sec:conclusion}.  Section~\ref{sec:sparse} shows (under some stronger assumptions) how to handle the case where each community is a sparse random graph.

\subsection{ Related Work}

The above setting seems superficially similar to other planted graph problems that were successfully treated using  SVD (singular value decomposition)
(see McSherry~\cite{McSherry01} and others). This similarity is however illusory because, first, the non-community edges in our model are
not necessarily randomly distributed, and more seriously, because the SVD techniques are known for finding
{\em vertex partitions} whereas here
due to overlap between communities we need to find {\em edge partitions}.
%we will explain below in Section~\ref{sec:previouswork} that since we need to allow the dense subgraphs to overlap substantially, the SVD-based %methods do not seem to suffice.

Eppstein, L\"{o}ffler, and Strash~\cite{EppsteinLS10} show how to provably find all maximal cliques in
time that is exponential in the ``degeneracy" of the network, which is bounded by the maximum degree in the worst case.
%Eppstein and Strash~\cite{EppsteinS11} find that the algorithm runs well in practice on up to a few million nodes.
Our network model, however, allowed graphs with arbitrary degeneracy and also works for concepts of community more general than a clique.
%  The degeneracy of a graph is the smallest number $m$ such that every subgraph contains a node with degree at most $m$, and in particular is bounded by the maximum degree.

Mishra et al.~\cite{MishraSST07}  also study overlapping communities in social networks.  They show a simple and elegant algorithm for detecting overlapping communities in a certain parameter space.  Their algorithm works best for communities where the overlap is not too large.  In our parameter space, to detect a community $C$ with density $\alpha$ and gap $\epsilon$,  they require that some $v \in C$ has fewer than $(\alpha+\epsilon-1)|C|$ neighbors outside of $C$. This is a strong restriction of the amount of overlapping, and
% which is not realistic for small $\epsilon$ and anyways
is impossible for small $\epsilon$ when $\alpha$ is bounded away from 1.

%\footnote{for more details, see section~\ref{sec:previouswork}}.

\paragraph{Related independent work: } Balcan et al.~\cite{BBBCT}
 have independently studied the problem of infering overlapping communities.
They have a very different starting point in terms of an explanatory model of how network ties are formed via a preference/ranking function among the
individuals. Surprisingly, they ultimately arrive at very similar set of ``minimal'' assumptions and algorithms for infering the communities. Perhaps this convergence is some kind of validation of both approaches.

\section{When Communities have Similar Sizes}
\label{sec:similarsize}

This section will give very efficient algorithms to find all communities in the graph when the following assumption holds:

%\medskip
{\em Assumption: Each community has size between $\delta k$ and $k$
where $\delta >0$ is some constant and $k$ is arbitrary but known to the algorithm.}

We continue to make the three assumptions made in Section~\ref{subsec:assumptions}, and the parameters $n, \gamma, d, \epsilon$ are as defined there.  To emphasize, the communities can be arbitrary sets in the graph so  long as the gap assumption is not violated and each node is in at most $d$ communities.  Furthermore, the placement of ``ambient" (i.e., non-community) edges in the graph can be adversarial as long as it does not violate the assumptions.
%Moreover, the gap assumption, the assumption that communities are cliques, and the assumption that communities are roughly %the same size will all be relaxed in future models.

 The running time is exponential in
$1/\delta$, so one would not use these algorithms if communities have radically different sizes; that case is handled in
Section~\ref{sec:differentsize}.

\subsection{Warmup: Communities are Cliques}
\label{sec:cliques}
In this section, ``community'' is understood to be synonymous with ``clique'' which corresponds to all affinity values $p_u=1$.

The algorithm focuses on the neighborhood $\Gamma(v)$ of a node $v$, and takes a random sample of nodes
$S$ from it. Then it uses brute force (or any other suitable heuristics) to find cliques of size about $\log d$  in the graph induced on $S$, and tries to extend them to communities. (To use sociology terms, here
{\em egocentric} or node-based analysis leads to provably correct {\em socio-centric} or societal analysis.)
The running time is linear in $n\cdot k$, albeit with a big
``constant'' factor term dependent upon $d,\delta, \epsilon, \gamma$.

\begin{theorem}
\label{thm:fastsamesizeclique}
%Given a Model A graph with parameters $(n, k, d, \delta, \epsilon, \gamma)$ for $k \geq 3$,
Given a graph satisfying the assumptions in this section, the {\sc Clique-Community-Find-Algorithm}
%below
outputs each community with probability at least $2/3$ in time \footnote{In this paper $\tilde{O}$ hides polynomial terms of the parameters $\delta$, $\gamma$, $\epsilon$ and also $\beta$, $\alpha$ if they are relevant.}
%$O(nkd^{Q\log R d})$ where $Q, R$ depend upon $\delta, \gamma, \epsilon$.
 %O(nk/\delta\gamma) (d\log(12d/\epsilon\delta\gamma)/\epsilon\delta\gamma)^{2\log(12d/\epsilon\delta\gamma)/\delta\epsilon} =
$O(nk/\delta\gamma) 2^{\tilde{O}(\log^2 d)}$. %(12d/\epsilon\delta\gamma)^{\frac{3d}{\gamma \delta \epsilon}}$.
\end{theorem}

%Now we describe the algorithm.
% and interlace it with italicized claims that are easily checked to hold with high probability
%via a simple application of Chernoff bounds.
The intuition behind why sampling works is that
for any node $v$ the neighborhood $\Gamma(v)$  has size at most $kd/\gamma$ as communities have size at most $k$ .
Each community  $C$ containing $v$ lies within $\Gamma(v)$ and has size at least $\delta k$, which is at least a $\delta \gamma/d$ fraction of
$\Gamma(v)$. Thus a random sample of $\Gamma(v)$ will have many representatives from $S$.
The only subtlety is to watch out for ``false positives'': sets that are not cliques but may present themselves as such during sampling.

\medskip

\noindent{\sc Clique-Community-Find-Algorithm}
\begin{algorithmic}[1]
\STATE Pick $\frac{9n}{\delta k}$ ``starting nodes'' uniformly at random and do the following:
\STATE For each starting node $v$ randomly sample $S \subseteq \Gamma(v)$  by including each node with probability
$p = \log(12d/\epsilon\delta\gamma)/\delta \epsilon k$. Proceed only if the sample has size
at most three times the expectation $\frac{deg(v)\log(12d/\epsilon\delta\gamma)}{\delta \epsilon k} \le \frac{d \log(12d/\epsilon\delta\gamma)}{\gamma \delta \epsilon} $.
\FORALL{cliques $U$ of size at most $2pk$ in the induced graph $G(S)$ on $S$ } \label{line:clique:fastsamesize:selectclique}
\STATE Let $V'$ be the set of nodes in $\Gamma(v)$ which are connected to all nodes in U, and let $G(V')$ be the induced subgraph on $V'$ \label{line:clique:fastsamesize:almostclique}
\STATE Let $U'$ be the set of nodes in $G(V')$ whose  degree in this subgraph is at least $(1-\epsilon/2)|V'|$. Output $U'$ if it is a clique of size at least $\delta k$, and for all $v\not\in U'$, $v$ is connected to at most a $1-\epsilon$ fraction of $U'$. \label{line:clique:fastsamesize:trim}
\ENDFOR
%\ENDFOR
\end{algorithmic}
%\medskip

%Clearly, the algorithm only outputs maximal cliques of size at least $\delta k$, where every node outside the clique is not connected to $\epsilon$
%fraction of the nodes in it. Under our assumptions each one of these
%corresponds to a community. We show that all communities are found with high probability.

%Thus $S \cap C$ will be one of the cliques in the induced graph on $S$, and it is simple to extract $C$ from it
%%y choosing nodes that have edges to at least  $1-\epsilon/2$ fraction of nodes in this clique.
%This almost works except for a few ``noisy'' nodes that have to be cleaned up.
%Now we present the full algorithm.

\begin{proof}(Theorem~\ref{thm:fastsamesizeclique})
 For any community $C$ the probability that a randomly chosen starting node $v$ belongs to it is at least $\delta k/n$. Thus
the expected number of times we pick a starting node from $C$ is at least $9$ and so by a Markov bound such a node is selected with probability $1/9$. We show that in each such trial the probability that community $C$ is output is at least $7/9$.

So let $v \in C$. Simple calculation based upon Chernoff bounds shows that each of the following sequence of three statements
holds with high probability.
(i) the subsampling gives a sample $S$ of size at most thrice the expectation.
%step~\ref{line:clique:fastsamesize:subsample} succeeds.
% (actually it is much smaller because of Chernoff bound).
(ii) $1/\epsilon \log (12 d/\epsilon\gamma\delta) \leq |S \cap C| \leq 2pk$.  Note that $S \cap C$ is a clique of $G(S)$.
Now  consider what happens when the for-loop of step~\ref{line:clique:fastsamesize:selectclique} tries $U = S \cap C$.
Since every node in $C$ has an edge to every node in $U$, the set $V'$ will contain $C$.
 (iii) $|V'| \leq |C| + \epsilon |C|/4$. This follows since by the Gap Assumption that
 each $u \in \Gamma(v) \setminus C$ has edges to at most a
$(1-\epsilon)$ fraction of nodes in $C$, and thus the probability that it has edges to each node in the random subset
$S \cap C$ is less than $\epsilon$. Also, the size of $\Gamma(v)$ is at most $kd/\gamma$, so in expectation the number of nodes in $V'\backslash C$ is only $(1-\epsilon)^{|S\cap C|}\cdot kd/\gamma \le \epsilon |C|/12$.

%we can select $U$ to be the intersection of $C$ and remaining nodes (these nodes definitely form a clique because $C$ itself is a clique). Clearly all nodes in $C$ will be in $V'$ (step~\ref{line:clique:fastsamesize:almostclique}) because $C$ is a clique. We shall show that $V'$ essentially only contains nodes in $C$.

\iffalse
For any individual $u\in \Gamma(v)\backslash C$, by property 5 of Model A we know it is only connected to at most a $(1-\epsilon)$ fraction of nodes in $C$. There are $\epsilon |C|$ nodes in $C$ that are not connected to $u$, each of them is removed with probability $1-p$, thus the probability that all of them are removed is

$$(1-p)^{\epsilon|C|} \le (1-p)^{\epsilon\delta k}  \le \epsilon\delta\gamma/12d.$$

If not all of these nodes are removed, then $u$ cannot be in the set $V'$ generated at step~\ref{line:clique:fastsamesize:almostclique}. Thus the expected size of $V'\backslash C$ is only an $\epsilon\delta\gamma/24d$ fraction of the degree of $v$. But the degree of $v$ is at most $kd/\gamma$, thus $\E[|V'\backslash C|] \le \epsilon \delta k /12 \le \epsilon |C|/12$. Again by Markov's Inequality we know the probability that $|V'\backslash C| \ge \epsilon|C|/4$ is at most $1/3$.
\fi

However, if $|V'\backslash C| \le \epsilon|C|/4$, then we can identify nodes in $C$ just by their degree in $G(V')$!
% in $V'$ as in step~\ref{line:clique:fastsamesize:trim}.
Each $w\in C$ has degree is at least $|C| \ge (1-\epsilon/2) |V'|$; whereas each  $w\not\in C$ has at most
% $1-\epsilon$ fraction of edges to $C$  its degree is at most
$|C|(1-\epsilon)+|V'\backslash C| \le (1-3\epsilon/4)|C| < (1-\epsilon/2) |V'|$.
%Therefore in this case (which happens with probability $1/3$) the algorithm indeed finds the community $C$.
Thus the algorithm returns exactly $C$.
\iffalse
Again by union bound we can show that each clique is found with probability $1-2^{-n}$. To bound the running time of the algorithm, clearly everything besides the two loops runs in time linear in the number of edges, the outer loop has $3n^2/\delta k$ iterations, the inner loop has at most $2^{\frac{d \log(12d/\epsilon\delta\gamma)}{\gamma \delta \epsilon}}$ iterations. Thus the running time would be $O(|E|n^2/\delta k)(12d/\epsilon\delta\gamma)^{\frac{3d}{\gamma \delta \epsilon}}$.
\fi
\end{proof}

\noindent{\bf A practical note:}
In step~\ref{line:clique:fastsamesize:selectclique} we enumerate over all cliques of a certain size. With a slight parameter modification we can show it suffices to enumerate over all maximal cliques of this size, for which in practice one may be able to use existing heuristic algorithms and reduce running time.

Namely, in the proof we pick $p =  \frac{\log(24d\log (d/\epsilon\delta\gamma)/\epsilon\delta\gamma)}{\delta \epsilon k}$, and still take $U$ to be $S\cap C$. Then Chernoff bound and union bound show that the probability that there is a node in $S\backslash C$ that connects to every node in $S\cap C$ is small. So in this case $S\cap C$ is a maximal clique in $G(S)$. %Enumerating over all maximal cliques is generally much faster than enumerating over all cliques in practice.

\subsection{When Communities are Dense Subgraphs}
\label{sec:dense}
%{\sc omit this discussion now in light of Section~\ref{subsec:assumptions}?}

In the previous section, we equated ``community'' with ``clique'', as has been done in many previous works. This assumes that everybody knows everybody else in a ``community''---clearly  a strong assumption as networks get larger (or even in smaller networks where data about adjacencies is incomplete).

In this section, we model a community as a dense subgraph $G(k, \alpha)$ which corresponds to all affinity values $p_u=\sqrt{\alpha}$. All sets have the same affinity $\sqrt{\alpha}$; though this will be relaxed later. If nodes $u, v$ are in
more than one community, their affinity is still the same and $e_{u,v} = \alpha$.
% in the next algorithm.

%{\sc sanjeev's note: do we need to define Model B now? I prefer to minimize the amount of assumptions.}

 The description of the algorithm will use the following notion, which every community necessarily satisfies.

%define $(\alpha, \alpha-\epsilon)$ set, which is an easy-to-check necessary condition for communities:

\begin{definition} \label{def:alphaepsilon} For $\alpha, \epsilon>0$ an $(\alpha, \alpha-\epsilon)$-set is a subset of nodes such that 1) every node in the set has edges to at
least an $\alpha$ fraction of nodes in the set; and 2) every outside node has edges to at most an $\alpha-\epsilon$ fraction of nodes in the set.
\end{definition}

\begin{theorem}
\label{thm:samesizedense}
Given a graph satisfying the conditions of this section with  %a Model B graph $G$ with parameters $(n, k, d, \alpha, \delta, \epsilon, \gamma)$ for
$k \gg \log n$,  \\the Dense-Community-Find-Algorithm outputs each community in $\mathcal{C}$ with probability at least \\$1-\exp(-\Omega(\alpha^2\epsilon^2 \delta k))$ over the randomness of $G$ and $2/3$ over the randomness of the algorithm  in time $O(n \cdot (k/\gamma\alpha\delta)\cdot 2^{\tilde{O}(\log^2 d)})$.%${O(\frac{2d \log(30d/\alpha\epsilon\delta\gamma)}{\alpha^2\delta\epsilon^2\gamma })})$.

%Let graph $G$ with $|V| =n$ be a random graph generated as follows: There is a matrix $E= \{e_{i,j}\}$ that specifies the probability of each edge being present, each undirected edge $(i,j)$ is present in $G$ with probability $e_{i,j}=e{j,i}$ and these events are independent. The matrix $E$ satisfies: 1) each node is in at most $d$ communities (so that $e_{i,j}$ when $i$, $j$ are in the same community is $\alpha$),  2) each community has between $\delta k$ and $k$ members, 3) at least a $\gamma$ fraction of each individuals expected number of edges are from community membership, 4) and if an individual $v$ is not a member of some community $C$, then the sum of $e_{u,v}$ for $u\in C$ is at most a $\alpha-\epsilon$ fraction of $|C|$ .  We also assume that $d \geq 2$ and that $k \gg \log n$. Then Clique-Community-Find-Algorithm will output each community of $G$ with probability at least $1-\exp(-\Omega(\alpha^2\epsilon^2 \delta k))$, and will run in time $\frac{3n^2}{\delta k}(O(|E|n) + (k d/\gamma)^{\frac{12d}{\gamma\alpha\delta\epsilon^2}})$.
\end{theorem}

\begin{proof}
We consider the following algorithm:

{\sc Dense-Community-Find-Algorithm}

\begin{algorithmic}[1]
\STATE Randomly choose $100n/\delta k$ nodes as starting nodes, for each starting node $v$ repeat the following
\STATE Let $G(\Gamma(v))$ be the induced subgraph of $v$ and her neighbors.
\STATE Subsample this subgraph by including each node with probability $p = \frac{2 \log(30d/\alpha\epsilon\delta\gamma)}{\alpha^2\delta\epsilon^2 k}$. Fail this round if this graph has size more than three times the expectation.
\FORALL{sets $U$ of nodes in the subsampled graph of size at most $2pk$} \label{line:dense:samesize:selectset}
\STATE Let $V'$ be the set of nodes in $\Gamma(v)$ which are connected to at least an $\alpha-\epsilon/2$ fraction of all nodes in U, let $G(V')$ be the induced subgraph on $V'$
\STATE Let $U'$ be the set of nodes in $G(V')$ such that their degree is at least $(\alpha-\epsilon/2)|V'|$ \label{line:dense:samesize:filter}
\STATE Let $U''$ be the set of nodes that has more than an $\alpha-\epsilon/2$ fraction of edges in to $U'$. Keep $U''$ if it is a $(\alpha-\epsilon/8,\alpha-7\epsilon/8)$ set
\ENDFOR
\end{algorithmic}

To analyze the algorithm, we first assume the graph $G$ is ``well formed''. The graph $G$ is well formed if 1) the number of edges from any node $v$ to any community $C$ is within $1\pm \epsilon/8$ of expectation. 2) For any node $u$, and node $v\in C$ in community $C$, the expected size of $\Gamma(u)\cap\Gamma(v)\cap C$ is within $1\pm\epsilon/8$ of expectation. In particular, we know if $u\in C$, $|\Gamma(u)\cap\Gamma(v)\cap C|\ge (\alpha-\epsilon/8)|\Gamma(v)\cap C|$; if $u\not\in C$, $|\Gamma(u)\cap\Gamma(v)\cap C|\le (\alpha-7\epsilon/8)|\Gamma(v)\cap C|$. %every community $C$ forms a $(\alpha-\epsilon/8, \alpha-7\epsilon/8)$ set. 2) the degree for any node does not exceeds two times its expectation. 3) for any two individuals $u$, $v$ in a community $C$, the number of their common neighbors in $C$ is at least $(\alpha-\epsilon/8)$ fraction of the number of neighbors of $v$ in $C$. That is, $|\Gamma(u)\cap\Gamma(v)\cap C|\ge (\alpha-\epsilon/8)|\Gamma(v)\cap C|$. 4) if $v$ is in community $C$ and $u$ is not in $C$, then the number of their common neighbors in $C$ is at most $(\alpha-7\epsilon/8)$ fraction of the number of neighbors of $v$ in $C$, $|\Gamma(u)\cap\Gamma(v)\cap C|\le (\alpha-7\epsilon/8)|\Gamma(v)\cap C|$.

Since all the requirements of well-formedness are $\epsilon$ far from their expectation, by Chernoff bound it is easy to show that the graph $G$ is well-formed with probability at least $1-\exp(-\Omega(\alpha^2\epsilon^2 \delta k))$.

Conditioned on $v\in C$ being selected, by Chernoff bounds we show the following statements hold with high probability: (i) the subsampling gives a sample of less than 3 times the expectation. (ii) If we choose $U$ to be the intersection of $\Gamma(v)\cap C$ and subsampled nodes, the number of edges from most nodes in $\Gamma(v)$ to $U$ will be close to expectation. (iii) The symmetric difference of $V'$ and $\Gamma(v)\cap C$ is at most $\epsilon|\Gamma(v)\cap C|/10$, because for any node in $\Gamma(v)$ to be in the symmetric difference its number of edges to $U$ will have to be $\epsilon/2$ away from expectation. Chernoff bound shows the probability of this is at most $\alpha\epsilon\delta\gamma/30d$, and $|\Gamma(v)| \le kd/\gamma$.

%If $v\in C$ is selected, we bound the probability that $C$ is detected. The subsampling step fails with probability at most $1/3$. Assume that does not happen, we take $U$ to be the intersection of the subsampled nodes and the community $C$. Notice that $U$ is also a uniform random subset of $C$. If $u\in C$, then the expected number of edges from $u$ to $U$ is $\alpha |U|$; if $u\not \in C$, the expected number of edges from $u$ to $U$ is at most $(\alpha-\epsilon)|U|$. For a node to be incorrect in $V'$ (meaning its in $V'$ but not $\Gamma(v)\cap C$ or vice verse) these numbers have to be %By Chernoff bound these numbers of edges will not be far from
%For a node $u$ to be in the symmetric difference of $\Gamma(v)\cap C$ and $V'$, its number of edges to the set $U$ has to be
% $\epsilon/2$ away from expectation. By Chernoff bound the probability is at most $\exp(-(\epsilon/2)^2 \alpha |U|)  \le \alpha\epsilon\delta\gamma/30d$. Since $|\Gamma(v)| \le kd/\gamma$, the symmetric difference of $V'$ and $\Gamma(v)\cap C$ has expected size at most $\epsilon|\Gamma(v)\cap C|/30$. With high probability the actual size is smaller than 3 times the expectation $\epsilon|\Gamma(v)\cap C|/10$.
Since the symmetric difference is so small, and $G$ is well formed, there will be a gap in degree for nodes in and outside $C$. For any $u\in |\Gamma(v)\cap C|$ the number of edges into $V'$ is at least $(\alpha-\epsilon/8-\epsilon/10)|\Gamma(v)\cap C|$; for any $u\not \in C$, the number of edges into $V'$ is at most $(\alpha-7\epsilon/8+\epsilon/10)|\Gamma(v)\cap C|$. Hence setting threshold at $\alpha-\epsilon/2$ suffices to distinguish the two cases. The set $U'$ is indeed a subset of $\Gamma(v)\cap C$ of size at least a $(1-\epsilon/10)$ fraction.

Finally, since the graph is well formed, any node $u\in C$ must have at least an $\alpha-\epsilon/8-\epsilon/10$ fraction of edges to $U'$, and any node $u\not\in C$ must have at most an $\alpha-7\epsilon/8+\epsilon/10$ fraction of edges to $U'$, again a threshold of $\alpha-\epsilon/2$ is enough to distinguish the two cases and $U'' = C$.

%For any $u\in \Gamma(v)\cap C$, as assumed we have $|\Gamma(u)\cap\Gamma(v)\cap C| \ge (\alpha-\epsilon/8)|\Gamma(v)\cap C|$, then the expected number of edges from $u$ to $V'\cap C$ is also at least $(\alpha-\epsilon/8)$ fraction of $|V'\cap C|$. The expected size of $V'\cap C$ is at least $p \alpha \delta k$, thus the probability that the fraction drops down to $(\alpha-\epsilon/4)$ is at most $\exp(-p\alpha\delta\epsilon^2 k) = (kd/\gamma)^{-2}$, since $\Gamma(v)\le 2dk/\gamma$ by union bound the probability that 2) is violated is small. Similarly the probability that 3) is violated is also small.

%As in Theorem~\ref{thm:samesizeclique} we know the probability that $C$ is not found conditioned on $v\in C$ is at most 2/3
%In conclusion, every time some $v\in C$ is selected we have a good probability of finding $C$, again we apply union bound to show that all communities are found with probability at least $1-2^{-n}$ (conditioning on the fact that $G$ is well formed).

The running time depends on the size of the subsampled nodes, which is of order $O(p\cdot kd/\gamma) = O(\frac{2d \log(30d/\alpha\epsilon\delta\gamma)}{\alpha^2\delta\epsilon^2\gamma })$. Thus the running time is $O(n (k/\alpha\gamma\delta)\cdot O(\frac{2d \log(30d/\alpha\epsilon\delta\gamma)}{\alpha^2\delta\epsilon^2\gamma })^{2pk}) = O(n \cdot (k/\alpha\gamma\delta) \cdot 2^{\tilde{O}(\log^2 d)})$.
\end{proof}

\subsubsection{Allowing Different Affinities}
In previous subsection we required edge probabilities $e_{u,v}$ to be exactly $\alpha$ if $u,v$ belong to the same community, and this probability does
 not rise even when they belong in more than one community.  % every edge of the graph $G$ to be independent. Also if $u, v$ belong to the same community, the probability that they are connected by an edge is exactly $\alpha$.
In real life these  requirements may be too stringent. Here we define the Dense-Similar-Size Assumptions which relax these two requirements.
In this model, the Dense-Community-Find-Algorithm may fail, and we give a new algorithm that, unfortunately, is less efficient.

\begin{model}{Dense-Similar-Size $(n, k, d, \alpha, \delta, \epsilon, \gamma)$}
%A graph $G$ with $n$ nodes and a set of communities $\mathcal{C}$ with maximum %overlap $d$ is consistent with Dense-Similar-Size Model  with parameters  if it satisfies

Communities satisfy Assumptions 1-3 from Section~\ref{subsec:assumptions} as well as the following:

%\medskip
$\star$ Each community $C\in \mathcal{C}$  has size between $\delta k$ and $k$
and is generated according to Assumption 1 with affinities $p_u \ge \sqrt{\alpha}$.%\label{modelA:CommIsClique} }

$\star$ If $u, v$ are in more than one community then edge has probability $e_{u,v}$ at least as large as the maximum requirement ($p_u p_v$) of all communities that they lie in.

\end{model}

\iffalse
{\sc This part talks about concentration, please move it to wherever it belongs}

The property~\ref{modelBstar:concentration} of Model B$^*$ is much weaker than the corresponding property of Model B.  Besides $G(k, \alpha)$ random graphs, several other graph generation models in the literature can be used to generate the edges within communities in such a way as to satisfy the models requirements, for example, the configuration model and the expected degree model~\cite{Jackson08}.  The first of these generates a multigraph with (nearly) any particular preassigned degree distribution.  The second generates a random graph where the expected degree of each node is (nearly) any preassigned value.  Both of these models can express the fact that some nodes have more participation with a community than other nodes, and thus we expect that they know more community members than do community members that participate with the community less.  This definition could also accommodate additional structure that introduces dependencies among the edges as long as there is still sufficient independence to satisfy property~\ref{modelBstar:concentration}.

{\sc End of this part}
\fi

%Again we stress that this model makes very few assumptions on the layout of %communities or the presents of ambient edges not within any community.

\begin{theorem}
\label{thm:robustsamesizedense}
Given a graph $G$ and a set of communities $\mathcal{C}$ consistent with Dense-Similar-Size Model with parameters $(n, k, d, \alpha, \delta, \epsilon, \gamma)$ where $d \geq 2$ and $k \gg \log n$, the {\sc Robust-Dense-Community-Find algorithm} below outputs each community in $\mathcal{C}$ with probability at least $1-\exp(-\Omega(\alpha^2\epsilon^2 \delta k))$ over the randomness of $G$ and probability at least $2/3$ over the randomness of the algorithm  in time $O(n \cdot (k/\alpha\delta\gamma)^{2\log(10/\epsilon)/\alpha+1}2^{\tilde{O}(\log^2 d)})$.%${O(\frac{d\log(10/\epsilon)\log(240\log(10/\epsilon)d/\alpha\epsilon\delta\gamma)}{\alpha^2 \delta \epsilon^2\gamma})}$.
\end{theorem}

\begin{proof}
The previous algorithm may fail because in this model $\Gamma(v)\cap C$ is no longer a uniform subset of $C$ and can be biased. Thus for a vertex $u$ the fraction of edges into the set $\Gamma(v)\cap C$ may be quite different from the fraction of edges into $C$. The idea of the algorithm is that for any community $C$, there is always a set $S$ such that $\Gamma(S)$ contains a large ($\ge 1-\epsilon/10$) fraction of $C$.
A uniform sample on $\Gamma(S)\cap C$ will be similar enough to a uniform sample on $C$, and the number of edges into sample will be close to the expectation. This allows us to get a set that is very close to $\Gamma(S)\cap C$ and then extend it similarly as before. % After subsampling, there will be $O(\frac{\log(120Td/\epsilon\delta\gamma)}{\alpha \epsilon^2})$ nodes in $C$ that gets sampled. Now in the for-loop (line~\ref{line:dense:robust:inducedgraph}) we can assume $U$ to be the intersection of the subsampled nodes and $C$. These nodes are random over a large fraction of $C$, so by concentration most nodes outside $C$ should have fewer than $\alpha-3\epsilon/4$ fraction of edges to these nodes; while most nodes inside $C$ have more than $\alpha-\epsilon/4$ fraction of edges to these nodes. Therefore $V'$ should be a set that has very small symmetric difference with the community $C$. At step~\ref{line:dense:robust:trim} we trim the set $V'$ by a threshold on degree, so that only nodes within $C$ will remain inside, thus $U'$ will be a subset of $C$ that contains more than $1-\epsilon/10$ fraction of nodes. Finally $U''$ will be exactly equal to $C$.

{\sc Robust-Dense-Community-Find}
\begin{algorithmic}[1]
\STATE Let $T = 2\log(10/\epsilon)/\alpha$.
\STATE Randomly choose $100n/\delta k$ starting nodes, for each starting node $v$ repeat the following
\FORALL{sets of nodes $S\subseteq \Gamma(v)$ of size $T$}
\STATE Let $G(\Gamma(S))$ be the induced subgraph of $S$ and their neighbors. \label{line:dense:robust:inducedgraph}
\STATE Subsample this subgraph by including each node with probability $p = O(\frac{\log(120Td/\epsilon\delta\gamma)}{\alpha \delta \epsilon^2 k})$.  Fail this round if this graph has size more than three times the expectation.
\FORALL{sets $U$ of nodes in the subsampled graph of size at most $2pk$} \label{line:dense:robust:innerloop}
\STATE Let $V'$ be the set of nodes in $\Gamma(S)$ which are connected to at least an $\alpha-\epsilon/2$ fraction of all nodes in U, let $G(V')$ be the induced subgraph on $V'$
\STATE Let $U'$ be the set of nodes in $G(V')$ such that their degree is at least $(\alpha-\epsilon/2)|V'|$ \label{line:dense:robust:trim}
\STATE Let $U''$ be the set of nodes that has more than an $\alpha-\epsilon/2$ fraction of edges in to $U'$. Keep $U''$ if it is a $(\alpha-\epsilon/8,\alpha-7\epsilon/8)$ set
\ENDFOR
\ENDFOR
\end{algorithmic}

We call the graph $G$ {\em well formed} if the degree of each node and the number of edges from any node to any community are within $1\pm \epsilon/8$ multiplicative factor of their expectations, also for any $u,v\in C$ the size of their intersection in $C$ $|\Gamma(u)\cap\Gamma(v)\cap C|$ is within $1\pm \epsilon/8$ of the expectation. By concentration bounds and union bound, the probability that $G$ is well formed is at least $1-\exp(-\Omega(\alpha^2 \epsilon^2 \delta k))$. We shall assume $G$ is well formed in the discussions below.

%Now we follow the plan and prove a series of claims:
%We first show %\begin{claim}
For any community $C$, when some $v\in C$ is the starting node, %there is a set $S$ of size $T$ such that $\Gamma(S)$ contains at least $1-\epsilon/10$ fraction of nodes in $C$.
let $S$ be a random subset %This is by
%\end{claim}
%\begin{proof}
%randomly picking
of  $T$ nodes in $C\cap \Gamma(v)$. Since the size $|\Gamma(u)\cap \Gamma(v)\cap C|$ is concentrated for any $u,v\in C$ the probability that none of these $T$ nodes are adjacent to $u$ is at most $(1-\alpha+\epsilon/8)^T < \epsilon/10$. Thus the expected size of $\Gamma(S)\cap C$ is at least $(1-\epsilon/10)|C|$.
%\end{proof}

We fix a set $S$ such that $\Gamma(S)\cap C$ contains at least a $1-\epsilon/10$ fraction of $C$, and show that $C$ is found with good probability. With high probability the subsampling step returns a sample of size less than 3 times the expectation. %The probability that the subsampling step fails is at most $1/3$ by Markov's Inequality.
After sampling, fix $U$ to be %We prove the following claim that asserts $V'$ is very close to $C\cap \Gamma(S)$.
%\begin{claim}
%With probability at least $2/3$, there is a set $U$ such that the symmetric difference between $V'$ and $C\cap\Gamma(S)$ is at most $\epsilon|C|/20$.
%\end{claim}
%\begin{proof}
%Fix the set $U$ to be
the intersection of subsampled nodes and the community $C$. Then this $U$ is a uniform sample of the set $\Gamma(S)\cap C$. %(which as we proved has size at least $(1-\epsilon/10)|C|$).
For any node $v\in C$, the expected number of edges from $v$ to $U$ is at least $(\alpha-\epsilon/10-\epsilon/8)|U|$; for any node $v\not\in C$, the expected number of edges from $v$ to $U$ is at most $(\alpha-7\epsilon/8)|U|$. By Chernoff bound these values are $\epsilon/4$ away from expectation (and thus the node is in the symmetric difference of $V'$ and $C\cap \Gamma(S)$) with probability less than $\epsilon\gamma\delta/120Td$. The size of $\Gamma(S)$ is at most $2Tkd/\gamma$. With high probability the symmetric difference (of $V'$ and $\Gamma(S)\cap C$)has size smaller than $\epsilon |\Gamma(S)\cap C|/20$.

%Now we only need to show when we have a set $V'$ that has very small symmetric difference with $C\cap \Gamma(S)$, the $U'$ generated will be a large subset of $C$ and $U''$ will be exactly $C$.
Now since $V'$ is really close to $C\cap \Gamma(S)$, it is easy to check that for all vertices $u\in C\cap V'$, the degree in $V'$ is larger than $(\alpha-\epsilon/2)|V'|$; for all $u\in V'\backslash C$ the degree in $V'$ is smaller. %there is a gap in the number of edges to $V'$ between nodes that are in or outside $C$, % This is simple because if $v\in V'\backslash C$, $v$ can only be connected to at most $(\alpha-7\epsilon/8+\epsilon/20)|C|$ nodes in $V'$; however if $v\in V'\cap C$, $v$ must be connceted to at least $(\alpha-\epsilon/10-\epsilon/20)|C|$ nodes in $V'$,
Thus setting a threshold at $\alpha-\epsilon/2$ suffices to distinguish these two cases, it follows that $U' = V'\cap C$. %Since $U'$ is a subset of $C$ and has size at least $(1-\epsilon/10-\epsilon/20)|C|$, every node inside $C$ must be connected to at least $(\alpha-\epsilon/4-3\epsilon/20)|C|$ nodes in $U'$, every node outside $C$ can only be connected to at most $(\alpha-7\epsilon/8)|C|$ nodes in $U'$, thus
Now $U'$ is a large subset of $C$, all vertices $u\in C$ will have more than $(\alpha-\epsilon/2)|U'|$ edges to $C$ while all vertices $u\not\in C$ have less edges. Setting the threshold at $\alpha-\epsilon/2$ is again sufficient to distinguish the two cases and $U''=C$.

%In conclusion, if $S$ if fixed correctly every iteration has probability at least $1/3$ of detecting community $C$. There are $n$ iterations so the probability that a community is detected (conditioning on $G$ being well formed) is at least $1-2^{-\Omega(n)}$.

Finally, the running time of the algorithm depend on the size of the subsampled nodes, which is at most $2Tkd/\gamma \cdot p = O(\frac{Td\log(120Td/\epsilon\delta\gamma)}{\alpha \delta \epsilon^2\gamma})$. Thus the algorithm runs in time \\
$n (kd/\gamma\epsilon\delta)^{T+1} O(\frac{Td\log(120Td/\epsilon\delta\gamma)}{\alpha \delta \epsilon^2\gamma})^{2pk} = O(n \cdot  (kd/\gamma\epsilon\delta)^{2\log(10/\epsilon)/\alpha+1}\cdot 2^{\tilde{O}(\log^2 d)})$.
\end{proof}

Notice that although the algorithm works for only a fixed value of $\alpha$, if the communities have different densities we can also apply the algorithm with different $\alpha$ parameters to find all communities.

\section{When Communities may have Very Different Sizes}
\label{sec:differentsize}
When communities have very different sizes, the parameter $\delta$ for our models in Section~\ref{sec:similarsize} can be too small and the algorithms are not efficient. In this section we show we can relax the similar size requirement using a quasi-polynomial time algorithm. We can also find cliques of different sizes in polynomial time with some additional assumptions.

\subsection{Quasi-polynomial Time Algorithm for Communities of Different Sizes}

When we have quasi-polynomial time, we can find all communities that have at least constant density just using assumptions 1, 2 and 3. We only make sure that the minimum density we want to find is a constant $\alpha_{min}$, that is, each community satisfies Assumption 1 with smallest $p_u\ge \sqrt{\alpha_{min}}$.

\begin{theorem}
\label{thm:anysizedense}
Given a graph $G$ satisfying the assumptions above with parameters $(n, k, d, \alpha_{min}, \delta, \epsilon, \gamma)$, if all communities are $(\alpha_C-\epsilon/8, \alpha_C-7\epsilon/8)$ sets (which happens with high probability when the size of the communities are not too small) the {\sc Any-Size-Dense-Community-Find} algorithm will output all communities in $\mathcal{C}$   in time $n^{\frac{100\log (kd/\gamma)}{\alpha_{min}\epsilon^2}+3}$.
\end{theorem}

\begin{proof}
When trying to apply previous ideas to this model, the difficulty is that communities have very different sizes and sampling will not find small communities. To solve the problem we just enumerate over all sets $S$ of size $T = \frac{100\log (kd/\gamma)}{\alpha_{min}\epsilon^2}$, think of all these points are chosen uniformly at random from a certain community $C$. This $S$ will serve as the sampled points, and since it is large we can apply union bound to show we will make no error when extending it to a community.

%\vspace{10pt}
{\sc Any-Size-Dense-Community-Find}
\begin{algorithmic}[1]
\STATE Let $T = \frac{100\log (kd/\gamma)}{\alpha_{min}\epsilon^2}$
\FOR{$\alpha = 1$ downto $\alpha_{min}$ step $-\epsilon/4$}\label{line:dense:anysize:loop}
\FORALL{sets of nodes $S$ of size $T$}
\STATE let $U$ be the set of all nodes that has more than $\alpha-\epsilon/4$ fraction of edges to $S$.
\STATE keep $U$ if it is a $(\alpha,\alpha-\epsilon/2)$ set.
\ENDFOR
\ENDFOR
\end{algorithmic}

For each community $C$ with density $\alpha_C$, there must be a value of $\alpha$ in the loop (line~\ref{line:dense:anysize:loop}) where $\alpha_C \ge \alpha+\epsilon/8$ and $\alpha_C-\epsilon<\alpha-\epsilon/2 - \epsilon/8$ (because the stepsize is $\epsilon/4$).
Assume this is the case, and let $S$ be a uniformly random set of size $T$ in $C$. For any node $v$, if $v\in C$, then the expected number of edges to the set $S$ is more than $\alpha$ fraction; if $v\not\in C$ the expected number of edges to the set $S$ is less than $\alpha-\epsilon/2$ fraction. The probability that the number of edges are $\epsilon/4$ fraction away from expectation is at most $\exp(-(\epsilon/4)^2 T\alpha_{min}) \leq (kd/\gamma)^{-2}$. We only need to apply union bound on the nodes of $C$ and their neighbors, so the size is much smaller than $(kd/\gamma)^2$. By union bound the probability that the algorithm successfully find $C$ is not 0. Since we are trying all possible sets $S$ the algorithm will always find all the communities.
\end{proof}

Although the algorithm is for dense subgraphs, if run it with $\alpha_{min}=1$, it will find all clique communities of any size.
%
%(We can also run the algorithm until all the edges are accounted for, or even choose nodes that are incident to edges not accounted for.)
%
%
%\begin{wrapfigure}{r}{0.45\textwidth}
%%  \vspace{-20pt}
%  \hrule\medskip
%  \begin{algorithmic}[1]
%    \STATE \textbf{let} $X_1 \gets X$
%
%    \FOR{$i = 1$ to $2$}
%
%    \STATE \textbf{let} $S_i \gets$ Greedy$(X_i)$
%    \label{step:greedy}
%
%    \STATE \textbf{let} $S'_i \gets$ \FMV$_\alpha(S_i)$
%    \label{step:fmv}
%
%    \STATE \textbf{let} $X_{i+1} \gets X_i \setminus S_i$.
%
%    \ENDFOR
%
%    \RETURN best of $S_1, S_1', S_2$. \label{step:ret}
%  \end{algorithmic}
%  \medskip\hrule
%  \caption{Submod-Max-Cardinality$(X,k,f)$}
%  \label{alg:card}
%  \vspace{-20pt}
%\end{wrapfigure}
%

\subsection{Polynomial Time Algorithm for Cliques of Different Sizes}

%We consider relaxing the assumption in Model A that the communities have
%sizes between $\delta k$ and $k$.
Now we try to improve the quasipolynomial time in Theorem~\ref{thm:anysizedense}
in the subcase when communities are cliques of different sizes.
The idea will be to reduce the amount of sampling and exhaustive enumeration.
% communities can overlap or even contain each other,
To prove this works we need to make assumptions beyond $1, 2, 3$.  %An immediate difficulty is that we cannot use subsampling directly: the sample will never touch the smaller communities. Even worse, s
The difficulty is that the solution can be highly
nonunique and degenerate if communities are allowed to be
too  ``similar.''  For example, suppose  node $w$ is not in a community $C$ but is contained in other communities
with large subsets of $C$. Should we now consider $w$ to be part of $C$, since it does have edges to all (or most) of $C$?
Our network model assumes such cases do not arise.
% and
 %that communities are ``different'' in a strong sense: each node $v$ in each community $C$ has $\beta$ fraction of its community edges going to nodes that
%do not lie in any other community $C' \neq C$ that contains $v$.

%To make sure these unreasonable cases do not happen, we make further assumptions:

%\medskip
\noindent{\em Assumption 4) Communities are fairly distinct.}
For each node $u$ in community $C$, at least a constant factor, say $\beta$, of $C$ does not lie in any other community containing $u$.
This is in accord with the intuitive view of how communities arise: the interconnection structure provides utility to its members
above and beyond what existed before~\cite{Feld81}.

%This also seems a fair assumption from a sociological viewpoint, since communities are believed to arise in networks because they provide utility to their members above and beyond what is obtainable in existing communities.

The next assumption is technical and perhaps was being assumed by the reader all along.
Surprisingly, we did not need it until now.

%\medskip
\noindent{\em Assumption 5) Completeness Assumption.} Any set that satisfies all the assumptions of a community in the model is a community.   (Also called ``Duck Assumption'': ``{\em If it looks like a duck, quacks like a duck, and walks like a duck, it's a duck.}'') This  ensures the adversary can't  satisfy Assumption 4
by just pretending that a certain set  is not a community even though it looks like one.

Finally we want to strengthen Assumption 3 so that smaller communities are distinguishable in principle from the noise introduced by the ambient edges:

%\medskip
\noindent{\em Assumption 3')} Every community a node $v$ belongs to has size at least $\gamma/d$ times the number of ambient edges incident on the node $v$.

\begin{theorem}
\label{thm:anysizeclique}
Given a graph that satisfies all assumptions in this section with parameters \\$(n, m, k, d, \beta, \epsilon, \gamma)$ where that $k \geq 3$, the {\sc Any-Size-Clique-Community-Find-Algorithm} will output all communities in $\mathcal{C}$ with probability at least $1-n^{-5}$ in time \\ $O(n\log n \cdot (kd/\gamma)^{\log(2/\epsilon)/\beta+1} \cdot 2^{\tilde{O}(\log^2 d)})$.%2^{O( \frac{Td \log(30T d/\epsilon\gamma)}{\epsilon \gamma})})$ where $T$ is $1/\beta\log (2/\epsilon)$.

%Let graph $G$ with $|V| =n$ be such that 1) each node is in at most $d$ communities (which are cliques),  2) member of each community $|C|$ has only this community in common with at least a $\beta |C|$ other members of the community, 3) for any individual, the size of the smallest community he's in is at least a $\gamma/(1-\gamma)$ fraction of the number of edges he has that are not from community membership , 4) and if an individual $v$ is not a member of some community $C$, then $v$ knows at most a $1-\epsilon$ fraction of the members of $C$ .  5) All communities are of size at least $m$ and all cliques of size at least $m$ are communities.

%The Any-Size-Clique-Community-Find-Algorithm will output each community of $G$ with probability at least $1-2^{-\Omega(n)}$, and will run in time $(nd\log (2/\epsilon)/\gamma\beta)^{O(\frac{d\log 2/\epsilon}{\epsilon\gamma\beta})}$.
\end{theorem}

\begin{proof}
The main algorithmic difficulty will be that $\Gamma(v)$ for any node $v$ may contain cliques of many different sizes. A subsample of $\Gamma(v)$ would be likely to hit
large cliques quite often, but not the smaller cliques. To solve this problem we try to find large cliques first. After cliques of size greater than $k$ are found, we can henceforth ignore their edges, and proceed to find smaller cliques.

Another problem is that after removing all edges in the large cliques, the remaining neighborhood of $v$
(called $\Gamma^-(v)$ in the algorithm)
 may not contain all nodes of a remaining clique. To solve the problem the algorithm uses a set $S$ of size $T$. We should think of $S$ as a random set in the community $C$, then by Assumption 4 and concentration bounds we know a large fraction of nodes in $C$ are in $\Gamma^-(S)$.

%We call this ``Model A'".

%{\sc Rong: Here we override the meaning of $\delta$, maybe we should use some other letter?}

%We consider the following algorithm: %when the smallest community can have size $m$:

{\sc Any-Size-Clique-Community-Find-Algorithm}

%\vspace{10pt}
%Let $k = n$. Repeat the following $\log n$ times:
%
%\begin{enumerate}
%\item Randomly choose a node $v$
%
%\item Let $G(\Gamma(v))$ be the induced subgraph of $v$ and her neighbors.
%
%\item Let $T = 2 \log (1/\delta)$ and enumerate over all subsets of nodes of size T:  $S \subseteq \Gamma(v)$ such that $|S|= T$.  (We are hoping that $S \subseteq C$ where $C$ is a community of size between $k/2$ and $k$).
%\item Let $\Gamma^-(S)$ be the set of nodes that are connected with some node in $S$ using an edge that does not belong to any of the communities already found by the algorithm. Consider $G(\Gamma^-(S))$, the induced graph of the $T$ nodes in $S$ and their ``out of community'' neighbors. We denote the size of this new induced graph $L$.
%\item Subsample this subgraph to expected size $\frac{2 \log(k(d-1)/\gamma)}{\delta \epsilon}$ by including each node with probability $p = \frac{2 T \log(k(d-1)/\gamma)}{\delta \epsilon k}$  Fail this round if this graph has size more than three times the expectation.
%
%\item Find all maximal cliques in subsampled graph.
%
%\item Extend all cliques greedily and keep those that are of size at least $\delta k$.
%
%\end{enumerate}
%Let $k = k/2$

% Repeat the following $\log n$ times:
\begin{algorithmic}[1]
\STATE Let $l = k$.
\WHILE{$l \ge m$} \label{line:clique:anysize:whileloop}
\STATE Randomly choose $100n\log n/l$ starting nodes , repeat the following for each starting node $v$ \label{line:clique:anysize:pickvertex}
\STATE Let $G(\Gamma(v))$ be the induced subgraph of $v$ and her neighbors.
\STATE Let $T = \log (2/\epsilon)/\beta$ and enumerate over all subsets of nodes of size T:  $S \subseteq \Gamma(v)$ such that $|S|= T$.  (We are hoping that $S \subseteq C$ where $C$ is a community of size between $l/2$ and $l$).
\FORALL{ choices of $S$}
\STATE Let $\Gamma^-(S)$ be the set of nodes that are connected with some node in $S$ using an edge that does not belong to any of the communities already found by the algorithm. Consider $G(\Gamma^-(S))$, the induced graph of the $T$ nodes in $S$ and their ``out of community'' neighbors. We denote the size of this new induced graph $L$.
\STATE Subsample this subgraph by including each node with probability $p = \frac{4 \log(30T d/\epsilon\gamma)}{\epsilon l}$  Fail this round if this graph has size more than three times the expectation.
\FORALL{Cliques $U$ in the subsampled graph of size at most $2pl$}
\STATE Let $V'$ be the set of nodes in $\Gamma(v)$ which are connected to all nodes in $U$, and let $G(V')$ be the induced subgraph on $V'$.
\STATE Let $U'$ be the set of nodes in $G(V')$ whose degree in this subgraph is at least $(1-\epsilon/4)|V'|$. Greedily extend $U'$ to a maximal clique $U''$. Output $U''$ if it is a clique and for all $v\not\in U''$, $v$ is connected to at most $1-\epsilon$ fraction of $U''$. \label{line:clique:anysize:greedyextension}
\ENDFOR
\ENDFOR
\STATE Let $l = l/2$
\ENDWHILE
\end{algorithmic}

%The algorithm tries to find large cliques first.
We show that if the algorithm correctly finds all cliques of size larger than $l$, then an iteration of the WHILE loop at line~\ref{line:clique:anysize:whileloop} will correctly find all communities with size between $l/2$ and $l$ with probability $1-n^{-10}$. The theorem then follows from a union bound.

Fix a community $C$ of size between $l/2$ and $l$, and assume a node $v\in C$ has already been chosen at step~\ref{line:clique:anysize:pickvertex}. Let $S$ be a random subset of $\Gamma(v)$ of size $T$. By Assumption 4 we know even after ignoring all edges of larger size, the number of remaining edges from any $u\in C$ to $C$ is at least a $\beta$ fraction, thus $|\Gamma^-(u)\cap C|\ge \beta |C|$. A random set of size $T$ intersects any set of size $\beta|C|$ with probability %The probability that none of these $\beta|C|$ points are inside $S$ is only $(1-\beta)^T =
$1-\epsilon/2$, therefore in expectation $\Gamma^-(S)$ contains a $1-\epsilon/2$ fraction of $C$. Since we are enumerating over all sets $S$ we can now assume $S$ is such that $\Gamma^-(S)$ contains at least $1-\epsilon/2$ fraction of $C$.

Now similar to Theorem~\ref{thm:fastsamesizeclique} it is easy to check that the following statements hold with high probability (i) the subsampling gives a sample of size at most thrice the expectation. (ii) For any node $u\not\in C$ there is a set of size at least $\epsilon|C|/2$ in $\Gamma^-(S)\cap C$ that is not connected to $u$. Suppose we take $U$ to be the intersection of community $C\cap \Gamma^-(S)$ and subsampled nodes, then we have (iii) $|V'|\leq |\Gamma^-(S)\cap C| + \epsilon|C|/15$. This is because by Assumption 3' the size of $\Gamma^-(S)$ is bounded by $Tld/\gamma$ and each of the nodes outside $C$ has only $\exp(-p\epsilon|C|/2)$ probability of being in $V'$.

The last event implies $V'$ is really close to $\Gamma^-(S)\cap C$. Now in graph $G(V')$ each $u\in V'\backslash|C|$ has degree $(1-\epsilon+\epsilon/15)|C|$; each $u\in C$ has degree at least $(1-\epsilon/2)C$. This gap enables the algorithm to use a threshold of $(1-\epsilon/4)|V'|$ to distinguish whether $u$ is in the community $C$ or not. The set $U'$ will be equal to $\Gamma^-(S)\cap C$.%, which is at least $(1-\epsilon/2)$ fraction of $C$.

%Now we do the subsampling and by Chernoff bound the probability that the algorithm fails at this step is very low. After subsampling let $U$ be the intersection of subsampled nodes and the community $C$. For any node $u\not\in C$ not in the community $C$, it can only be connected to $(1-\epsilon)$ fraction of nodes in $C$, thus there is a set of $\epsilon/2$ nodes in $\Gamma^-(S)\cap C$ that does not have an edge to $u$. The probability that every node in $U$ is adjacent to $u$ is therefore $(1-p)^{\epsilon|C|/2} \le \epsilon/Td30$. However, the number of nodes in $\Gamma^-(S)$ is at most $Tkd/\gamma$, thus the number of nodes outside $C$ in $V'$ is only $Tkd/\gamma \cdot \gamma\epsilon/Td30\le \epsilon k/30 \le \epsilon |C|/15$. Clearly if $u\in V'\backslash|C|$ it can only be connected to $(1-\epsilon+\epsilon/15)|C|$ nodes in $V'$, and if $u\in C$ it is connected to at least $(1-\epsilon/2)|C|$ nodes in $V'$, the size of $V'$ is between $(1-\epsilon/2)|C|$ and $(1-\epsilon/2+\epsilon/15)|C|$, thus setting the threshold at $(1-\epsilon/4)|V'|$ suffices to distinguish whether $u$ is in the community $C$ or not. The set $U'$ will be equal to $\Gamma^-(S)\cap C$, which is at least $(1-\epsilon/2)$ fraction of $C$.

Finally by the Gap Assumption we know during the greedy extension of step~\ref{line:clique:anysize:greedyextension}, we can only include nodes in $C$ and in fact will include all nodes in $C$. Therefore $U'' = C$ and the community $C$ is found with high probability.

The running time of the algorithm is dominated by the round when $l =k$. At that round on the size of the subsampled set is at most $O(p \cdot Tkd/\gamma) = O( \frac{Td \log(30T d/\epsilon\gamma)}{\epsilon \gamma})$, we want to find a set of size $2pk$. Thus the running time is \\$O(n\log n  (kd/\gamma)^{T+1}O( \frac{Td \log(30T d/\epsilon\gamma)}{\epsilon \gamma})^{2pk} = O(n\log n (kd/\gamma)^{\log(2/\epsilon)/\beta+1} 2^{\tilde{O}(\log^2 d)})$.

\end{proof}

We leave it as an open problem to identify reasonable set of assumptions that allow polynomial time when communities are dense subgraphs.
The problem is that the ``duck assumption" is not well defined: we know what a clique looks like, but it is hard to tell whether a subgraph "looks like" a community generated according to Assumption 1 when there are overlapping communities and ambient edges. We could try to make a stronger "duck assumption" by assuming every large $(\alpha, \alpha -\epsilon)$- set is a community, and then a similar algorithm will be able to find all dense communities in polynomial time. But this is not as reasonable as our other assumptions: consider two
$(\alpha, \alpha -\epsilon)$-sets $C_1$ and $C_2$ of size $2k/3$ and their intersection has size $k/3$, then it's quite likely that their union $C_1 \cup C_2$ is a  $(\alpha, \alpha -\epsilon)$ set but we don't consider this set as a community.

\iffalse
%we can also hope for finding all communities in polynomial time. However,
The problem is that the ``duck assumption'' is not well defined: we know what a clique looks like, but it is hard to tell whether a subgraph ``looks like'' a community generated as Assumption 1 when there are overlapping communities and ambient edges. We could try to make a stronger ``duck assumption'' by assuming every large $(\alpha,\alpha-\epsilon)$ set is a community, and then a similar algorithm will be able to find all dense communities in polynomial time. But this is not reasonable as our other assumptions: consider two $(\alpha, \alpha-\epsilon)$ sets $C_1$ and $C_2$ of size $2k/3$ and their intersection has size $k/3$, then it is quite likely that their union $C_1\cup C_2$ is a $(\alpha/2, \alpha/2-\epsilon)$ set but we don't consider this set as a community.

\fi

\section{Relaxing the Assumptions}
\label{sec:real-life}
%\label{sec:real-world}

\subsection{Relaxing Assumption 1}

\label{subsec:relax1}
Assumption 1 states that each community's edges are generated according to a expected degree model. In this section we note that the algorithms and proofs of Theorems~\ref{thm:robustsamesizedense} and Theorem~\ref{thm:anysizedense} actually apply to a more general setting.

We first note that we can substantially relax the Dense-Similar-Size Model by replacing Assumption 1 with the following two requirements:
%\begin{description}

\noindent \textbf{Concentration}: the number of edges from any node $u$ to any community $C$ is concentrated around the expectation, that is, $\Pr[|\Gamma(u)\cap C| \not\in[(1\pm \epsilon) \E[|\Gamma(u)\cap C|]]] \le \exp(-\epsilon^2\E[|\Gamma(u)\cap C|])$, and the degree of each node is concentrated similarly.

\noindent \textbf{$(\alpha, \epsilon)$-Regularity}: for all $u,v\in C$,
   $\Pr[|\Gamma(u)\cap \Gamma(v)\cap C| \leq [(1 - \epsilon) \E[\alpha |\Gamma(v)\cap C|]]]$ \\ \indent  \hspace{110mm} $\le \exp(-\epsilon^2\E[\alpha |\Gamma(v)\cap C|])$
%\end{description}

These properties do not require full dependence, but only limited independence, which could be satisfied, for example, by the configuration model ~\cite{Jackson08} which generates a multigraph with (nearly) any particular preassigned degree distribution.  This definition could also accommodate additional structure that introduces dependencies among the edges as long as there is still sufficient independence to satisfy Concentration and Regularity.  Consider, for instance, the disjoint union of two equal-sized cliques with a random bipartite graph of density $\beta$ between them.  This is $\alpha$-regular for any $\beta > \frac{\alpha}{4 - 2 \alpha}$.  Thus communities can be much more clumpy than in the expected degree model.

\begin{remark}
Theorem~\ref{thm:robustsamesizedense} still holds with the same proof after replacing Assumption 1 with Concentration and $(\alpha, \epsilon)$-regularity.
\end{remark}

%%Could put this in conclusion???
%A key assumption we made was that the communities themselves were $\alpha$-regular which implies they are not too clumpy--any two people in a community have at least some expected number of mutual neighbors within the community.  Intuitively, this assumption means that the ego-centric network of a node $i$ in a community $C$ gives a fairly undistorted view of $C$.  We noted that this held in several graph generation models such as $G(n, \alpha)$, the configuration model, and the expected degree model.  However, these models are rather simplistic.  We feel the largest assumption that this paper makes is what the communities that we are trying to locate look like and we discuss this further below.

%We secondly note that we can relax Assumption 1 so that theorem~\ref{thm:anysizedense} %does not require the expected degree model but only requires concentration.
\begin{remark}
Theorem~\ref{thm:anysizedense} still holds with the same proof after replacing Assumption 1 with Concentration.
\end{remark}

%%%%%%%%%%%%%%%%%%

\subsection{Relaxing Assumption 2--the Gap Assumption}
\label{subsec:relax2}

Though
plausible, the Gap Assumption may not exactly hold in a real-life graph since there may be nodes that fall in the ``gap.'' Our algorithm needs to still return sensible answers.
 Now we argue that our algorithms in Theorem~\ref{thm:fastsamesizeclique}, and Theorem~\ref{thm:samesizedense} produce sensible answers even when this happens.
We use the Clique-Community-Find-Algorithm (Theorem~\ref{thm:fastsamesizeclique}) to illustrate.
Of course in this setting we cannot hope to return the exact communities. Instead the algorithm will return  some $C'$ that contains more than a $1-\epsilon$ fraction of $C$ and has density more than $1-\epsilon$.

%We restate in this setting.

\begin{theorem}
\label{thm:relaxgapA}
If $G$ is a graph that satisfies Assumptions 1 and 3, and each community is a clique that has size between $\delta k$ and $k$ where $\delta > 0$ is some constant and $k$ is arbitrary but known to the algorithm, then the Clique-Community-Find-Algorithm can be adapted so that for any community $C$, the algorithm finds a set $C'$ such that $|C'\cap C|\ge (1-\epsilon)|C|$ and for each $v\in C'$, the number of edges to $C'$ is at least $(1-\epsilon)|C'|$.
\end{theorem}

\begin{proof}
The idea is to run the Clique-Community-Find-Algorithm as before.  Once we get $V'$ corresponding to setting $U = S \cap C$ (recall that under the assumptions of Theorem~\ref{thm:fastsamesizeclique} that this $V'$ contains $C$ and has only $\frac{\epsilon}{4}|C|$ nodes outside $C$), we know with high probability $V'$ consists of 3 parts: the community $C$ itself, some set of nodes that have more than $1-\epsilon$ edges to $C$, and some set of nodes that have no more than $1-\epsilon$ edges to $C$.  In these 3 parts, the first part is what we want, the third part is very small by the proof of Theorem~\ref{thm:fastsamesizeclique}, and so only the second part worries us. Among these three parts, the nodes in $C$ should tend to have the largest degrees, so we will use the degree of these nodes to identify them.
 For any node $v$ in $V'$, we call $|\Gamma(v)\cap V'| /|V'|$ the {\em density} of $v$.
The idea will be to  run Clique-Community-Find-Algorithm with some parameter $\epsilon'$ that is somewhat smaller than $\epsilon$ to get $V'$. Then repeatedly remove nodes from $V'$ that have density less than $1-\epsilon/2$  until the density of all remaining nodes is more than $1-\epsilon$. We will show by a simple calculation that not many nodes in $C$ will be removed. The details are as follows:

\begin{enumerate}
\item{Run Clique-Community-Find-Algorithm with $\epsilon' = \epsilon^2 / 6\log (d/\delta\gamma)$. }
\item{Once the algorithm has produced the set $V'$, repeatedly remove all nodes of density less than $(1-\epsilon/2)$ until for every $v\in V'$, the density is at least $(1-\epsilon)$.}
\end{enumerate}

Let $V' = C\cup H \cup W$, where $C$ is the community, $H$ are the nodes that are connected to more than a $1-\epsilon'$ fraction of $C$ and $W$ are the nodes that are connected to at most a $1-\epsilon'$ fraction of $C$. By the proof of Theorem~\ref{thm:fastsamesizeclique} $|W|\le \epsilon' |C|$, thus it is very small and can be ignored in the computation below.

The size of $V'$ is at most $kd/\gamma$, because each node in it is  adjacent to the node $v$ in the algorithm. We shall show that (i) each iteration removes at most an $\epsilon/\log (d/\delta\gamma)$ fraction of nodes in $C$, and (ii) if any iteration removes less than half of the nodes, each node in the remaining graph will be connected to at least a $1-\epsilon$ fraction of nodes. Claim (ii) implies that we will have at most $\log (d/\delta\gamma)$ iterations and then Claim (i) implies that at most an $\epsilon$ fraction of nodes in $C$ are removed.

For (i), notice that as long as less than half of $C$ is removed, the fraction of edges from any node in $H$ to the remaining part of $C$ is at least $1-2\epsilon'$. Thus the average density of nodes in $C$ is at least $1-3\epsilon'$. By Markov's inequality the fraction of nodes that have density less than $1-\epsilon/2$ is at most $3\epsilon'/(\epsilon/2) = \epsilon/\log (d/\delta\gamma)$.

For (ii), notice that all remaining nodes had density $1-\epsilon/2$, if less than half of the nodes are removed, their density should still be larger than $1-\epsilon$.

Notice that since the $\epsilon'$ used in {\sc Clique-Community-Find-Algorithm} is now $\epsilon^2/6 \log(d/\delta \gamma)$, the running time of the algorithm will be $O(n k/\delta\gamma 2^{\tilde{O}(\log^3 d)})$.
\end{proof}

Similar ideas can be used to relax the gap assumption in Theorem~\ref{thm:samesizedense}. The main difficulty of applying the same argument is that the edges not from community membership can be adversarially chosen. However in real life graphs this is not likely to happen: if two people do not share any community the probability that they know each other should be lower. If for all edges $(u,v)$ we have the probability $e_{u,v} \le \alpha$, then for any community $C$ the following algorithm can always find some set $C'$ with density at least $\alpha-\epsilon$ that contains at least a $1-\epsilon$ fraction of $C$:

\begin{enumerate}
\item{Run Robust-Dense-Community-Find algorithm with $\epsilon' = \epsilon^2 / 10\log (d/\delta\gamma)$. }
\item{Once the algorithm has produced the set $V'$, repeatedly remove all nodes of density less than $\alpha-\epsilon/2$ until the density of every node is at least $(\alpha-\epsilon)$.}
%\item{Find $t \in \{0,1,...,1/\epsilon'\}$ such that the number of nodes with density in range $[\alpha-\epsilon/2+t\epsilon'^2, \alpha-\epsilon/2+(t+1)\epsilon'^2]$ is smaller than $\epsilon'$ fraction. Remove all nodes of density less than $\alpha-\epsilon/2+(t+1/2)\epsilon'^2.$}
\end{enumerate}

The proof is very similar to Theorem~\ref{thm:relaxgapA}. First we focus only on the expected degree of nodes. Since $e_{u,v}\le \alpha$ we can normalize these probabilities by multiplying by $1/\alpha$. Now $e_{u,v} = 1$ for nodes in a community, and $e_{u,v}\le 1$ for all pairs $u$, $v$. Thus same argument as Theorem~\ref{thm:relaxgapA} shows that the algorithm works if we are given the true values of $e_{u,v}$. Then we argue that the algorithm should also work even if we are just given a random graph $G$, because the algorithm only uses the degree of nodes in various sets and they all concentrate around their expectation. There are some techinalities here when many nodes have expected degree very close to the threshold we are setting, which can be resolved if we choose a random threshold between $\alpha-\epsilon/2$ and $\alpha-0.4\epsilon$.

\section{Sparser Communities}
\label{sec:sparse}
In previous sections we have been talking about communities as dense graphs. This is natural when the community considered is small and people inside are closely related, such as people in the same year and department in a university. However this may not be true for larger communities: if we consider all students in a large university, or even all computer scientists, then it is unlikely that every person knows a constant fraction of other people in the community. In this section we show how our ideas can be applied to communities that are not so dense.

Consider a simple set of assumptions (``Sparse'') where the affinities for a community $p_u = \Omega(k^{-1/4})$. That is two people in the same community of size $k$ know each other only with probability $\Omega(1/\sqrt{k})$.
We assume the network satisfies Assumption 1 from Section~\ref{subsec:assumptions} as well as the following:

%\begin{enumerate}
%\item{$G$ is a random graph such that the edge $(u,v)$ present with some probability $e_{u,v}$ independently.}
\noindent{\bf (1.)} Each community has size $k$, for any two nodes $u$, $v$ in the same community the probability $e_{u,v} = B/\sqrt{k}$ where $B$ is a constant larger than 10.  %}
%\item{Each individual is in at most $d$ communities}
%\item{If two nodes $u$, $v$ share some community, then $e_{u,v} = B/\sqrt{k}$ where $B$ is a constant larger than 10.}
\noindent{\bf (2.)} There are no ambient edges. %}% If two nodes $u$, $v$ do not share any community, then $e_{u,v} = 0$.}
\noindent{\bf (3)} The intersection of any two communities $C\cap C'$ has size at most $k/20d^2$. %\label{modelC:smallintersection}}
%\end{enumerate}

\iffalse
\begin{model}{Sparse}
A graph $G$ with $n$ nodes and a set of communities $\mathcal{C}$ with maximum overlap $d$ is consistent with the Sparse Model with parameters $(n, k, d, B)$ if it satisfies Assumption 1 from Section~\ref{subsec:assumptions} as well as the following:

%A graph $G$ with $n$ nodes and a set of communities $\mathcal{C}$ are generated by Model C with parameters $(n, k, d, B)$ if they satisfy the following properties:

\begin{enumerate}
%\item{$G$ is a random graph such that the edge $(u,v)$ present with some probability $e_{u,v}$ independently.}
\item{Each community has size $k$, for any two nodes $u$, $v$ in the same community the probability $e_{u,v} = B/\sqrt{k}$ where $B$ is a constant larger than 10.
%\item{Each individual is in at most $d$ communities}
%\item{If two nodes $u$, $v$ share some community, then $e_{u,v} = B/\sqrt{k}$ where $B$ is a constant larger than 10.}
\item{There are no ambient edges.%}% If two nodes $u$, $v$ do not share any community, then $e_{u,v} = 0$.}
\item{The intersection of any two communities $C\cap C'$ has size at most $k/20d^2$. \label{modelC:smallintersection}}
\end{enumerate}
\end{model}
\fi
Notice that we do not need to require the gap assumption nor the duck assumption here because they are both implied by property~3
%\ref{modelC:smallintersection}
and the fact that every edge is in some community.

Instead of giving an algorithm to find communities under the ``Sparse" assumptions, we show that it can be transformed to a graph $G'$ such that $G'$, $\mathcal{C}$ fulfill the Dense-Similar-Size Assumptions. Then we can directly apply the {\sc Robust-Dense-Community-Find} algorithm of Theorem~\ref{thm:robustsamesizedense} to find the communities.

\begin{theorem}
\label{thm:sparse}
Let a graph $G$ and a set of communities $\mathcal{C}$ be consistent with the Sparse Model. Construct a graph $G'$ on the same set of nodes, where $u,v$ has an edge in $G'$ if and only if they have at least $B^2/2$ length-2 paths in $G$. Then the pair $(G',\mathcal{C})$ are consistent with Dense-Similar-Size Assumptions with parameters $(n, k, d, \alpha, \delta, \epsilon, \gamma)=(n, k, d, 0.9, 1, 0.6, 1/3d)$.
\end{theorem}

\begin{proof}
%Among the properties of Model Dense-Similar-Size, property~\ref{modelBstar:size} and Limited Membership Property follows directly by definition. %3, 4 (each community has between $\delta k$ to $k$ members; each person is in at most $d$ communities) are straight forward by the properties 2 and 3.
We rely on the relaxation of the Dense-Similar-Size model in Section~\ref{subsec:relax1} using the concentration and ($\alpha$, $\epsilon$)-regularity.

For concentration, focus on one community $C$, notice that once we fix all the edges adjacent to $u$, the probability that $v$ has more than $B^2/2$ length-2 paths in $G$ are independent for different $v$'s in the same community. This is because the number of length-2 paths is completely determined by the number of edges from $v$ to $\Gamma_G(u)$, and these are disjoint sets for different $v$'s. Moreover, by symmetry the probability only depends on the number of edges adjacent to $u$ in community $C$. Thus once we fix the degree of $u$ inside $C$ in graph $G$, all the edges $(u,v)$ where $v\in C$ are independent and they satisfy a Chernoff bound. The degree of $u$ itself is also concentrated.

For ($\alpha$, $\epsilon$)-regularity, we show that for any $u,v\in C$, the size of their intersection within $C$ $|\Gamma(u)\cap \Gamma(v)\cap C|$ is also concentrated.  Just consider randomly choosing $\Gamma(u)$ and $\Gamma(v)$, with high probability both their sizes and the size of their intersection are close to the expectation. In this case whether some node $w$ has many length-2 paths to $u$ or $v$ is also independent (because the relevant edges are disjoint for different $w$). Chernoff bounds implies the concentration.

For the probability of edges $\alpha$, the expected number of length-2 paths between $u$ and $v$ in the same community $C$ is at least $|C|\cdot (B/\sqrt{k})^2 = B^2$, by Chernoff bound the probability that the number drops down to half of expectation is smaller than 0.1 when $B>10$.

For Assumption 3 in Section \ref{subsec:assumptions}, for each node $v$, if it is in $d'\le d$ communities, by the calculation above the expected number of community edges in $G'$ is at least $kd' \cdot 0.9 \cdot 0.9 > 0.8kd'$ (here the first 0.9 is the probability of an edge within the community, the second 0.9 is because the communities may overlap, however by property~3 the overlap must be small). The expected number of length-2 paths starting from $v$ in $G$ is at most $kd'\cdot (B/\sqrt{k} \times kd \cdot (B/\sqrt{k}) = kd'dB^2$, thus the expected number of edges of $v$ in $G'$ must be smaller than a $2/B^2$ fraction of this, which is $2kd'd$. Since $0.8kd'/2kd'd > 1/3d$ the number of ambient edges is small.

Finally, we would like to show the gap assumption: if $u\not \in C$, the expected number of edges from $u$ to $C$ in $G'$ is small. To do that we only need to show the expected number of length-2 paths from $u$ to $C$ in $G$ is small. We divide the length-2 paths from $u$ to $C$ into two cases: those that enter $C$ at the first step and those that enter at the second step. For the first type, the expected size of $\Gamma(u)\cap C$ is only $k/20d^2 \cdot d \cdot  B/\sqrt{k} = B\sqrt{k}/20d$, thus the number of length-2 paths from $u$ to $C$ where the first step is inside $C$ is at most $B\sqrt{k}/20d \cdot k \cdot B/\sqrt{k} = B^2k/20d$. For the second type, in the second step each node only has at most $k/10d^2 \cdot d \cdot  B/\sqrt{k} = B\sqrt{k}/20d$ expected edges to $C$, thus the total expected number of length-2 paths is at most $dk\cdot B/\sqrt{k} \cdot B\sqrt{k}/20d = B^2k/20$. Combining the two cases, we know the expected number of length-2 paths from $u$ to $C$ is at most  $B^2k/20+B^2k/20d \le B^2k/10$, thus the expected number of edges in $G'$ is at most $B^2k/10 \cdot 2/B^2 = 0.2 \le \alpha - \epsilon$.

\end{proof}

\section{Conclusion}
\label{sec:conclusion}

%We studied the problem of identifying communities in a graph representing a social network when each node may belong to more than one community,  but at most some constant $d$ %communities.  We
We introduced a framework for rigorously thinking about community structure that allows (a) overlapping communities (b)  includes well-studied notions such as cliques and
dense subgraphs as subcases  and (c) yet allows efficient algorithms for recovering the communities. Our assumptions lie between worst-case and average-case, are based on a long line of research, and we suspect they hold in many generative models.
%used traditional definitions of community and studied settings in which our algorithms could efficiently and provably recover the network structure.
%Our algorithms exploited the fact that the ego-centric networks were dense to import technology from the study of dense graphs. We hope that this provides a framework for future work.
Our sampling-based techniques infer global structure (socio-centric analysis) from the neighborhood of vertices (ego-centric analysis).  This local versus global framework,
familiar in computer science,
% of using relationships between local and global structures to parlay ego-centric network %analysis to an understanding of the entire social network
may be useful in other settings in sociology, especially because ego-centric networks are empirically observed to be dense and thus amenable to our techniques.
 We think our techniques should meld well with existing heuristics and plan to do a performance study on real-world data, and also to test
the validity of our assumptions.
% on it.
%Our framework also points to an interesting local-global relationship in social networks that is worth exploring: to what extent can you recover information about the network by looking %only/mostly at ego-centric networks?  How does this relationship hold up in practice versus theory?

%First, it would be interesting to implement our algorithms to see how well they perform on real-world data.  Such implementations could be aided by heuristics for dense subgraph and related %problems which are provably correct and attempt to have better average-case running times than the subroutines currently in our algorithms.  We note that such algorithms exist for finding %all maximal cliques~\cite{BronK73}.

\iffalse
Another direction would be to test how well our assumptions hold in actual data sets.  A key assumption we made was that the communities themselves were $\alpha$-regular which implies they are not too clumpy--any two people in a community have at least some expected number of mutual neighbors within the community.  Is this a sufficiently general assumption?  Can this assumptions safely be tightened.
\fi

Weakening our assumptions is another promising direction, and  we made a start in Section~\ref{sec:real-life}.  The Gap Assumption (Assumption 2) makes intuitive sense but probably cannot be guaranteed for all network nodes (eg, there will be an occasional node that knows more community members than some particular member, yet is not a community member).  Our use of the {\em expected degree model} for the intracommunity edges
(Assumption 1) can be weakened somewhat, but it still is a {\em static} model.
%is static, whereas

Arguably community evolution is a dynamic process that results in a more intricate, and possibly hierarchical structure.   Researchers have started considering two-step models. For example the first step could generate an initial graph according to our assumptions and in the second step each node connects to each neighbor of a neighbor with some small probability.
%(There are several simple variants of this where the second step is done iteratively, or, perhaps communities are better at facilitating new ties in the second step than are ambient edges.)
%Such a dynamic processes (whose study is still in its infancy) creates structure and %dependencies within and between communities which may violate our assumptions.
Making these models amenable to efficient community-detection
%xtending our results to this setting
is a good open problem.

\subsection*{Acknowledgements}
We would like to thank Nikhil Srivastava for contributions in the early part of this work that helped this project find its final direction. Thanks to Balcan et al. for giving us a manuscript of their independent work~\cite{BBBCT}.  We also would like to thank Bernie Hogan for useful consultations about the sociology literature.

\bibliographystyle{abbrv}

\bibliography{biblionew}

%\appendix
%
%\section{Omitted Proofs}
%\label{sec:appendix}
%\input{appendix.tex}

\end{document}